\documentclass[preprint,3p]{elsarticle}

\usepackage{color}
\usepackage{amssymb}
\usepackage{amsmath}
\usepackage{appendix}
\usepackage{xcolor}
\usepackage{multicol}
\usepackage{float}
\usepackage[T1]{fontenc}

\newtheorem{Proposition}{Proposition}

\newtheorem{Remark}{Remark}

\newtheorem{Ex}{Example}%[section]

\newtheorem{proofhead}{Proof}

\newenvironment{proof}
    {
        \par
        \begin{proofhead}
        \normalfont
    }
    {   \qed
        \end{proofhead}
        \par
    }
\biboptions{sort&compress}
\journal{}

\begin{document}

\begin{frontmatter}

%% Title, authors and addresses

%% use the tnoteref command within \title for footnotes;
%% use the tnotetext command for theassociated footnote;
%% use the fnref command within \author or \address for footnotes;
%% use the fntext command for theassociated footnote;
%% use the corref command within \author for corresponding author footnotes;
%% use the cortext command for theassociated footnote;
%% use the ead command for the email address,
%% and the form \ead[url] for the home page:
%% \title{Title\tnoteref{label1}}
%% \tnotetext[label1]{}
%% \author{Name\corref{cor1}\fnref{label2}}
%% \ead{email address}
%% \ead[url]{home page}
%% \fntext[label2]{}
%% \cortext[cor1]{}
%% \affiliation{organization={},
%%             addressline={},
%%             city={},
%%             postcode={},
%%             state={},
%%             country={}}
%% \fntext[label3]{}

\title{Testing and estimation of the index of stability of univariate and bivariate symmetric $\alpha-$stable distributions via modified Greenwood statistic}

%% use optional labels to link authors explicitly to addresses:
%% \author[label1,label2]{}
%% \affiliation[label1]{organization={},
%%             addressline={},
%%             city={},
%%             postcode={},
%%             state={},
%%             country={}}
%%
%% \affiliation[label2]{organization={},
%%             addressline={},
%%             city={},
%%             postcode={},
%%             state={},
%%             country={}}

\author[inst1]{Katarzyna Skowronek}
\author[inst2]{Marek Arendarczyk}
\author[inst3]{Anna K. Panorska}
\author[inst3]{Tomasz J. Kozubowski}
\author[inst1]{Agnieszka Wyłomańska}

\affiliation[inst1]{organization={Faculty of Pure and Applied Mathematics, Hugo Steinhaus Center, Wroclaw University of Science and Technology},%Department and Organization
            addressline={Hoene-Wronskiego 13c},
            city={Wroclaw},
            postcode={50-376},
            %state={State One},
            country={Poland}}
                       \affiliation[inst2]{organization={Mathematical Institute, University of Wrocław},%Department and Organization
            addressline={pl. Grunwaldzki 2 },
            city={Wrocław},
            postcode={50-384},
            %state={State One},
            country={Poland}}
\affiliation[inst3]{organization={Department of Mathematics \& Statistics, University of Nevada, Reno},%Department and Organization
            addressline={1664 N Virginia St},
            city={Reno},
            postcode={89557},
            %state={State One},
            country={USA}}
         \begin{abstract}
We propose a testing and estimation methodology for univariate and bivariate symmatric $\alpha$-stable distributions using a modified version of the Greenwood statistic. Originally designed for positive-valued random variables, the Greenwood statistic, and its modified version tailored for symmetric distributions, have been predominantly applied to univariate random samples. In this paper, we extend the modified Greenwood statistic to a bivariate setting and examine its probabilistic properties within the class of $\alpha$-stable distributions, with a focus on the sub-Gaussian case. Additionally, we introduce a novel testing approach that considers two variations of the modified Greenwood statistic as test statistics for the bivariate case. In the univariate setting, we adapt the proposed testing methodology for estimating the stability index. The simulation studies presented demonstrate that our proposed methodology outperforms classical approaches previously used in this context and serves as an effective tool for distinguishing between Gaussian and $\alpha$-stable distributions with a stability index close to 2. The theoretical and simulation results are further illustrated with practical data examples.
\end{abstract}

%%Graphical abstract
%\begin{graphicalabstract}
%\includegraphics{grabs}
%\end{graphicalabstract}

%%Research highlights
%\begin{highlights}
%\item Research highlight 1
%\item Research highlight 2
%\end{highlights}

\begin{keyword}
%% keywords here, in the form: keyword \sep keyword
$\alpha-$stable distribution \sep bivariate sub-Gaussian distribution \sep Greenwood statistic \sep Monte Carlo simulations \sep testing for finite variance \sep testing hypotheses
\end{keyword}

\end{frontmatter}

%@@
%@@
%@@
\section{Introduction}
\label{biv.test.green.intro}
%@@
%@@
%@@
The primary objective of this paper is to introduce a methodology for testing whether a univariate or bivariate random sample originates from an $\alpha$-stable distribution within the sub-Gaussian class, utilizing the recently developed modified Greenwood statistic \cite{skowronek2024}. Although the primary focus is on testing bivariate distributions, we also examine the univariate case to highlight the versatility of the modified Greenwood statistic in addressing the testing  problem across both dimensions.  In the univariate setting, we also adapt the proposed testing methodology for estimating the stability index.

We recall that the Greenwood statistic, first introduced by M. Greenwood in 1946~\cite{medicine} and further elaborated in~\cite{Moran_1947}, has been extensively studied in the literature. Over the years, many researchers have explored this statistic and its extensions from a theoretical perspective. For instance, variants of the Greenwood statistic have been applied in goodness-of-fit tests for exponential and uniform distributions (see, e.g., \cite{goodness-fit}). In~\cite{albrecher2007} and~\cite{albrecher2009}, researchers derived the asymptotic properties of the moments of the Greenwood statistic, while~\cite{Albrecher_2010} examined the asymptotic distribution behavior, focusing on the existence of moments in the underlying distribution. The stochastic monotonicity of the Greenwood statistic, under the assumption of star-shaped stochastic monotonicity in the underlying random sample, was proven in~\cite{arendarczyk2022}. Furthermore, the Greenwood statistic was applied in~\cite{arendarczyk2023} to develop a test for inferring the tail index of generalized Pareto distributions. Notable applications of the statistic include its use in testing Taylor's law, as discussed in~\cite{taylor-law1,taylor-law2}. More recently, a generalization of the Greenwood statistic, involving three additional parameters, was proposed in~\cite{ALBRECHER2022}, where the authors analyzed its asymptotic properties for regularly varying distributions. Another generalization, focusing on the asymptotic efficiency of the proposed statistic, was presented in~\cite{GeneralGreen2}.

The Greenwood statistic and its extended versions discussed above have found numerous applications across diverse fields. Notable examples include their use in analyzing clustering events in both space and time, particularly in fields like medicine and epidemiology~\cite{medicine,epidemology}, genetics and genomics~\cite{genetics,genomics}, biology~\cite{biology}, economics and insurance~\cite{economics,insurance}, hydrology~\cite{hydrology}, optimization~\cite{optimization}, physics and materials science~\cite{physics,materials}, anomaly detection~\cite{anomaly}, internet traffic monitoring~\cite{internet}, and even athletics~\cite{athletics}.

Recently, the authors of \cite{skowronek2024} introduced a modified Greenwood statistic that, unlike the classical version, can be applied to any real-valued random samples. One of the most intriguing properties of both the classical Greenwood statistic and its modified version, tailored for symmetrically distributed univariate samples, is their stochastic monotonicity within the class of $\alpha$-stable distributions. This characteristic allows the modified Greenwood statistic, as well as its classical counterpart, to serve as a reliable test statistic for hypotheses within this general class of distributions.

Recall that the univariate $\alpha$-stable distributions are defined by four parameters, with the stability index $\alpha \in (0, 2]$ being the most important. The symmetric $\alpha$-stable family is a two-parameter class, characterized by the stability index $\alpha$ and the scale parameter $\sigma$. A key feature of $\alpha$-stable distributions is that for $\alpha < 2$, they exhibit heavy tails, with the corresponding random variable having infinite variance. Additionally, $\alpha$-stable distributions generalize the Gaussian distribution, converging to it as $\alpha = 2$ (see, e.g.,~\cite{stable3}). For further details, we refer the readers to classical works on $\alpha$-stable signals and models, such as \cite{shao22,alek_book,non_gauss,Nolan2020}.

The $\alpha$-stable distribution class has been applied widely in various fields, including financial markets \cite{mand,fin_new1,fin_new2}, physics \cite{phys_new1,phys_new2,phys21,phys22,phys23}, biology \cite{biol_new1,biol_new2}, medicine (particularly in heartbeat analysis) \cite{phys35}, climate dynamics \cite{phys33}, telecommunications \cite{tel1,tel2}, and condition monitoring 
\cite{HMM:measurement,mon_new1,Zulawinski_MSSP}.

The natural generalization of the univariate symmetric $\alpha$-stable distribution is its multivariate version, which is characterized by the parameter $\alpha \in (0,2]$ (having the same role as in the univariate case) and the spectral measure $\Gamma$, which encodes the dependence structure between the components. In the univariate scenario, the spectral measure is concentrated at a single point, providing insights into the scale parameter $\sigma$ of the $\alpha$-stable random variable.

A special case of the multivariate symmetric $\alpha$-stable distribution is the class of sub-Gaussian distributions \cite{stable3}. These are defined as the product of the square root of a totally skewed $\alpha$-stable random variable and a Gaussian random vector with zero mean and an appropriate covariance matrix. In this case, the spectral measure is symmetric, and the dependence between the components of the sub-Gaussian vector is determined by the covariance matrix parameters.

Similar to the univariate case where symmetric $\alpha$-stable distributions generalize the Gaussian, sub-Gaussian distributions can be viewed as an extension of the multivariate Gaussian class, reducing to it when the stability index $\alpha = 2$. The class of sub-Gaussian random vectors has been widely studied in the literature (see, e.g., \cite{stable3,grzesiek2022,subGauss1,subGaussapp1}). These vectors have found applications in diverse fields, including finance \cite{subGaussapp1,subGaussapp2,subGaussapp3,subGaussapp4}, signal processing \cite{subGaussapp5}, filtering algorithms \cite{subGaussapp7}, hydrology \cite{subGaussapp6}, and biomedical data analysis \cite{subGaussapp8}.

The relatively simple form of the modified Greenwood statistic, as discussed in \cite{skowronek2024}, allows for a straightforward extension from the univariate to the multivariate domain. In this paper, we propose two such extensions and explore their properties for the class of sub-Gaussian distributions. While our focus is primarily on the bivariate case for simplicity, the methodology can be generalized to higher dimensions. A key aspect of our analysis is the stochastic monotonicity of the proposed statistics within the class of sub-Gaussian distributions. The probabilistic properties of the modified Greenwood statistic in the bivariate setting make it suitable for testing whether a given random sample originates from a sub-Gaussian distribution. Since sub-Gaussian distributions reduce to multivariate Gaussian distributions when the stability parameter $\alpha = 2$, these extensions are also useful for testing multivariate normality.

There are several well-known methods in the literature for testing the multivariate Gaussian distribution. These include Mardia's tests for skewness and kurtosis \cite{mardia}, the Jarque-Bera test, which combines Mardia's measures \cite{jb2test}, the Henze-Zirkler test based on differences between the empirical and theoretical characteristic functions \cite{henzezirkler}, and Royston's test \cite{royston}, a multivariate generalization of the Shapiro-Wilk test. Our proposed method adds to this body of work, and simulation studies show that it outperforms many classical approaches, especially for small sample sizes and when the sub-Gaussian distribution closely approximates the Gaussian case, i.e., when
$\alpha$ is near 2. In the univariate scenario, we extend the methodology to design confidence intervals for the stability parameter $\alpha$. The theoretical developments are complemented by simulation results and an illustrative real-world example.

%The rest of the paper is organized as follows. In Section \ref{sec:stable} we recall the definition of the $\alpha-$stable distribution in univariate and multivariate scenario paying the main attention on the class of sub-Gaussian random vectors. In Section \ref{sec:greenwood} we define two versions of the modified Greenwood statistic adopted for the bivariate random samples and discuss their properties within the class of sub-Gaussian distributions. In Section \ref{sec:testing} we present the testing methodology based on the discussed statistics in the univariate and bivariate scenarios. In the next section we present the simulation study demonstrating the efficiency of the proposed methodology in the statistical testing problem. In Section \ref{sec:real} we illustrate the presented methodology for real-world data. Last section concludes the paper.

The reminder of this paper is organized as follows. In Section  \ref{sec:stable}, we recall the definition of the $\alpha$-stable distribution in both the univariate and multivariate scenarios, with a particular focus on the class of sub-Gaussian random vectors. Section \ref{sec:greenwood} introduces two versions of the modified Greenwood statistic adapted for bivariate random samples and explores their properties within the class of sub-Gaussian distributions. In Section \ref{sec:testing}, we present the testing methodology based on the proposed statistics in both the univariate and bivariate settings.  In the univariate case we also present the estimation method for $\alpha$ parameter. Section \ref{sim:section} provides a simulation study that demonstrates the efficiency of the proposed methodology in statistical testing. In Section \ref{sec:real}, we illustrate the application of the proposed methodology using real-world data. The final section offers concluding remarks.

%@@
%@@
%@@
\section{The $\alpha$-stable distribution}
\label{sec:stable}
%@@
%@@
%@@
%{\color{red} Since we are not working with general univariate or multivariate stable distributions, I deleted the general definitions and left the definitions for the symmetric $\alpha$-stables and sub-Gaussian vectors.}

In this section, we introduce the families of univariate and multivariate symmetric $\alpha$-stable distributions. 
%@@
%@@
\subsection{Univariate $\alpha$-stable distribution}
%@@
%@@
%An important class of heavy-tailed distributions widely used in many practical problems is the class of the $\alpha$-stable distributions $S(\alpha,\beta,\sigma,\mu)$. A random variable from $\alpha$-stable distribution is defined through its characteristic function (CF) $\Phi_X(\cdot)$ given by the following formula \cite{stable3}
%
%\begin{equation}
%    \Phi_X(t) = \begin{cases}
%        \text{exp} \{ -\sigma^\alpha|t|^\alpha
% \{1 - i\beta \text{sign}(t) \text{tan} \frac{\alpha \pi}{2} \} +i \mu t \} & \text{ for } \alpha \neq 1 \\
 %       \text{exp} \{ -\sigma|t| \{ 1 + i\beta\text{sign}(t)\frac{2}{\pi}\text{log}(|t|)\} +i \mu t \} & \text{ for } \alpha = 1
%    \end{cases}.
%\end{equation}
%
%\noindent Here, $\alpha \in (0,2]$ is stability index - the most crucial parameter describing the distribution, $\beta \in [-1,1]$ is the skewness parameter,  $\mu \in \mathbb{R}$ is a shift parameter, and $\sigma > 0$ is the scale parameter. The stability index $\alpha$ dictates the heavy-tail behavior of the distribution; specifically,  the tail of the distribution follows a power-law decay, expressed as $1-F_X(x) \sim C x^{-\alpha}$, where $F_X(\cdot)$ is the cumulative distribution function (CDF) of $X$. Thus, the lower the stability index $\alpha$,  the more impulsive behavior within the random sample can be observed. It is noteworthy that the variance of the  $\alpha$-stable distribution does not exist when $\alpha \in (0,2)$. However, for $\alpha = 2$, the $\alpha$-stable distribution reduces to the Gaussian distribution, resulting in finite variance.\\
A symmetric $\alpha$-stable distribution, denoted by $S(\alpha,\sigma)$, is defined  by its characteristic function 
%\indent In this work, we specifically analyze the class of symmetric $ \alpha$-stable distributions $S(\alpha,\sigma)$, where $\beta = 0$ and $\mu = 0$. Thus, the CF describing the distribution simplifies to:
%
\begin{equation}\label{char}
    \Phi_X(t) =  \text{exp} \left( - \sigma^\alpha |t|^\alpha\right), \,\,\, t\in \mathbb R,
\end{equation}
where $\alpha \in (0, 2]$ is the tail parameter (stability index) and $\sigma > 0$ is the scale parameter. When $\alpha = 2$, the $\alpha$-stable distribution is the Gaussian distribution. When $\alpha \in (0,2)$ variance of $S(\alpha,\sigma)$ does not exist and the variability of a $S(\alpha,\sigma)$ random variable increases as $\alpha$ decreases. 
%@@
%@@
\subsection{Multivariate $\alpha$-stable vectors}
%@@
%@@
%The multivariate $\alpha$-stable distribution serves as a generalization of the univariate case to $d$-dimensions,  representing the distribution of a random vector $\mathbf{X} = (X_1,  \ldots, X_d)$. The distribution of the $\alpha$-stable random vector is defined through a CF given by:
%
%\begin{equation}
  %  \mathbf{\Phi}_{\mathbf{X}}(\mathbf{\theta}) = \begin{cases}
 %       \text{exp} \{- \int_{S_d} |\langle \theta , \mathbf{s} \rangle|^\alpha (1-i\text{sign}(\langle \theta , \mathbf{s} \rangle) \text{tan}(\frac{\alpha \pi}{2})) \Gamma(ds) + i \langle \sigma, \mu_0 \rangle \} &  \text{ for } \alpha \neq 1 \\
 %       \text{exp} \{- \int_{S_d} |\langle \theta , \mathbf{s} \rangle| (1+\frac{2}{\pi}i\text{sign}(\langle \theta , \mathbf{s} \rangle) \text{log}(\langle \theta , \mathbf{s} \rangle)) \Gamma(ds) + i \langle \sigma, \mu_0 \rangle \} & \text{ for } \alpha = 1,
 %   \end{cases}
%\end{equation}
%
%\noindent where $\mu_0 \in \mathbb{R}^d$ is a shift vector, $\langle \cdot, \cdot \rangle$ denotes the inner product,
%$\Gamma(\cdot)$ represents the spectral measure, and  $S_d$ is the unit sphere in $\mathbb{R}^d$. The spectral measure is a finite measure defined on $S_d$ and captures both the skewness and the dependency structure among the components of the random vector.\\

\indent In this work, we focus on sub-Gaussian distributions which are a special case of the multivariate $\alpha$-stable vectors.   A random vector $\mathbf{X} = (X_1, \ldots, X_d)$ from the sub-Gaussian distribution is defined as follows
\begin{equation}
\label{vector}
    \mathbf{X} = (X_1, \ldots, X_d) = A^{\frac{1}{2}} \mathbf{G} = (A^{\frac{1}{2}} G_1, \ldots, A^{\frac{1}{2}} G_d),
\end{equation}
\noindent where $\mathbf{G}$ is a standard (zero-one) Gaussian vector in $\mathbb{R}^d$ and $A$ is an independent from $\mathbf{G}$ $\alpha/2$-stable random variable with  $\sigma = (\text{cos}(\frac{\pi\alpha}{4}))^{2/\alpha}$. The characteristic function of a sub-Gaussian vector is 
%, $\beta=1$, and $\mu=0$. As proved in~\cite{stable3}, the distribution of the random vector from the sub-Gaussian distribution is characterized by the following CF:
%
\begin{equation}
    \mathbf{\Phi_X}(\mathbf{\theta}) = \text{exp} \left\{ -\frac{1}{2} \left| \sum_{i=1}^{d} \sum_{j=1}^{d} R_{ij} \theta_i \theta_j \right|^{\alpha/2 }\right\},
\end{equation}
\noindent where $R_{ij}=\rho_{ij}R_{ii}R_{jj}$ is the covariance between $G_i$ and $G_j$, while $R_{ii}=Var(G_i)$, and 
$\rho_{ij}$ denotes the correlation between $G_i$ and $G_j$. 
%For the sub-Gaussian distribution, the spectral measure
%$\Gamma$ is symmetric and assigns equal weights to antipodal sets on the sphere $S_d$. Moreover, if we assume that
%$\mathbf{G}$ in (\ref{vector}) is a Gaussian vector of independent components (i.e.,  $\rho_{ij}=0$ for $i\neq j$), the spectral measure is uniform  when $\alpha<2$ (see Proposition 2.5.5 in \cite{stable3}). Similar to the univariate case, for $\alpha=2$, the vector given in (\ref{vector}) becomes a Gaussian random vector. 
Note, that all univariate marginals $A^{\frac{1}{2}}G_i$ for $i=1,\ldots,d$ of the random vector $\mathbf{X}$ in (\ref{vector}), have 
univariate symmetric $\alpha-$stable distributions with  $\sigma=1/\sqrt{2}$ (see Proposition 1.3.1 in \cite{stable3}).

%Thus, in this case the dependence between the sub-Gaussian random variables $G_i$ and $G_j$ is known.
%@@
%@@
%@@
\section{Modified Greenwood statistic for sub-Gaussian random vectors}
\label{sec:greenwood}
%@@
%@@
%@@
In this section, we first recall the definition of the recently introduced modified Greenwood statistic \cite{skowronek2024}, and its main properties for univariate random samples coming from the symmetric $\alpha$-stable distribution. We then introduce two approaches for applying the discussed statistic to the analysis of multivariate random samples and discuss the properties of the extended versions of the modified Greenwood statistic for sub-Gaussian vectors. In this paper,  the analysis are provided for the bivariate case.

Let ${\mathbb{X}}=(X_1, X_2, \ldots, X_n)$ be univariate random sample - a sequence of independent and identically distributed (IID) random variables. Following \cite{skowronek2024}, we define a modified Greenwood statistic as follows
\begin{eqnarray}
\label{def_Greenwood}
        S_{\mathbb{X}} := \frac{\sum_{i=1}^n |X_i|^2}{\left(\sum_{i=1}^n |X_i|\right)^2}.
\end{eqnarray}
%
%The statistic defined in (\ref{def_Greenwood}) is scale invariant and stochastically decreasing with respect to the parameter $\alpha$ within the class of symmetric $\alpha-$stable distributions (see \cite{skowronek2024}).
The statistic defined in (\ref{def_Greenwood}) is scale invariant and stochastically decreases with respect to the parameter $\alpha$ within the class of symmetric $\alpha$-stable distributions (see \cite{skowronek2024}). This means that, for any constant scaling factor applied to the data, the value of the statistic remains unchanged, and as the stability parameter
$\alpha$ increases, the distribution of the statistic becomes stochastically smaller. This behavior aligns with the fact that symmetric $\alpha$-stable distributions become less heavy-tailed as $\alpha$ approaches 2, transitioning toward the Gaussian distribution.

We shall consider two statistics based on (\ref{def_Greenwood}) for testing hypotheses about the parameter $\alpha$ within the class of sub-Gaussian random vectors, focusing on the bivariate case (i.e. when $d=2$). A bivariate random sample coming from a sub-Gaussian distribution will be denoted by $\vec{\mathbb{X}}=(\mathbf{X}_1, \mathbf{X}_2,\ldots, \mathbf{X}_n)$ where $\mathbf{X}_i=\left(X_1^{(i)}, X_2^{(i)}\right)$, $i=1,\ldots,n$.

%In the next part, we define two  statistics that are based on statistic defined in Eq. (\ref{def_Greenwood}) and are further considered for testing  sub-Gaussian random vectors.  In this paper we restrict ourselves to the bivariate case (i.e., $d=2$). In the next subsections we use the notation $\vec{\mathbb{X}}=(\mathbf{X}_1, \mathbf{X}_2,\ldots, \mathbf{X}_n)$, where $\mathbf{X}_i=\left(X_1^{(i)}, X_2^{(i)}\right)$, $i=1,2,\ldots,n$, has sub-Gaussian distribution, for describing a  bivariate random sample coming from sub-Gaussian distribution.

%@@
%@@
\subsection{Bivariate case: the first scenario}
%@@
%@@
The first scenario is based on the statistic
\begin{eqnarray}
\label{def_Greenwood2}
    S_{\vec{\mathbb{X}}}^1=S_\mathbb{W},
\end{eqnarray}
where $\mathbb{W}=(W_1, \ldots, W_n)$ and $W_i=X_1^{(i)}+X_2^{(i)}$, $i=1,2,\ldots,n$.  Let us note that, for $\alpha=2$, the random variable $W=W_i$ has zero-mean Gaussian distribution with variance equal to $Var(X_1)+Var(X_2)+2\rho Var(X_1)Var(X_2)$, where $\rho$ is the correlation between $X_1=X_1^{(i)}$ and $X_2=X_2^{(i)}$. Thus, as the modified Greenwood statistic (\ref{def_Greenwood}) is scale invariant, the statistic $S_{\vec{\mathbb{X}}}^1$ does not dependent on $\rho$ and the variances of $X_1$ and $X_2$. Moreover, the same property holds for $\alpha<2$.  In this case, the random variable $W$ has a symmetric $\alpha-$stable distribution with parameter $\alpha$ and scale parameter $\sigma=\sqrt{1+\rho}$ (see Proposition 1.3.1. in \cite{stable3}). In both cases,   $\alpha=2$ and $\alpha\neq 2$, the random variables $W_1,W_2,\ldots, W_n$ are mutually independent. Thus, in this scenario, the statistic $S_{\vec{\mathbb{X}}}^1$ is also independent of $\rho$ and reduces to the modified Greenwood statistic for the univariate symmetric $\alpha$-stable random sample. Therefore, according to Proposition 3 in \cite{skowronek2024},  it is stochastically decreasing with respect to the parameter $\alpha$, as stated in the result below.
%@@
%@@
\begin{Proposition}
Let $S_{\vec{\mathbb{X}}}^1$ be the statistic defined in (\ref{def_Greenwood2}), where
$\vec{\mathbb{X}}=(\mathbf{X}_1,\mathbf{X}_2, \ldots, \mathbf{X}_n)$ is a random sample from a bivariate sub-Gaussian distribution with index $\alpha \in (0,2]$. Then $S_{\vec{\mathbb{X}}}^1$ is stochastically decreasing with respect to
$\alpha$.
\end{Proposition}
%@@
%@@

%@@
%@@
\subsection{Bivariate case: the second scenario}
%@@
%@@
In the second scenario, we consider the statistic
\begin{eqnarray}
\label{def_Greenwood1}
    S_{\vec{\mathbb{X}}}^2=S_\mathbb{Y},
\end{eqnarray}
where $\mathbb{Y}=(Y_1, \ldots, Y_n)$ and $Y_i=(X_1^{(i)})^2+(X_2^{(i)})^2$, $i=1,2,\ldots,n$.  Let us note that,  for $\alpha=2$, the random variable $(X_1,X_2)\stackrel{d}{=}(X_1^{(i)},X_2^{(i)})$ has  a bivariate Gaussian  distribution with mean zero and covariance matrix $\boldsymbol \Sigma=[R_{ij}]_{i,j=1,2}$, where $R_{ii}=Var(X_i)$, $i=1,2$,  and $R_{ij} = \rho \sqrt{Var(X_1)Var(X_2)}$ ($i\neq j$). Here,  $\rho\in [-1, 1]$  is the correlation between $X_1$ and $X_2$.  %Let us note that for a sub-Gaussian random vector with $\alpha=2$, the random variables $X_1$ and $X_2$ have the same variances. However, the following part of this section presents the discussion for the general case.

By spectral representation of the covariance matrix,
$\boldsymbol \Sigma = {\bf D}\mbox{diag}(\beta_1, \beta_2){\bf D}^\top$, where ${\bf D} = [{\bf d}_1, {\bf d}_2]$
is an orthogonal matrix of (column) eigenvectors ${\bf d}_1$, ${\bf d}_2$ and $\mbox{diag}(\beta_1, \beta_2)$ is a diagonal matrix of the corresponding eigenvalues $\beta_1$, $\beta_2$, we can express the random vector ${\bf X} = (X_1, X_2)$ as ${\bf X}\stackrel{d}{=} \beta_1 {\bf d}_1Z_1 +  \beta_2 {\bf d}_2Z_2$, where the $Z_1,Z_2$ are IID standard Gaussian variables. Thus, by the orthogonality of the $d_1,d_2$ and properties of the inner product, we have
\begin{equation}
\label{kret1}
Y = X_1^2 + X_2^2 = \langle {\bf X}, {\bf X} \rangle \stackrel{d}{=} \beta_1 Z_1^2+\beta_2Z_2^2,
\end{equation}
where the $Z_1^2,Z_2^2$ are IID $\chi^2$ random variables with one degree of freedom (and $Y$ has a generalized $\chi^2$ distribution, see \cite{MatPro92}). By factoring out the larger of the two eigenvalues in (\ref{kret1}), we obtain
\begin{equation}
\label{kret2}
Y \stackrel{d}{=} (\beta_1 \vee \beta_2) (\beta Z_1^2 + Z_2^2),
\end{equation}
where
\begin{equation}
\label{beta_def}
\beta  = \frac{\beta_1 \wedge \beta_2}{\beta_1 \vee \beta_2} \in [0,1]
\end{equation}
and the $\wedge$/$\vee$ are the $\min$/$\max$ operators.

Since the Greenwood statistic (\ref{def_Greenwood1}) is evaluated on the random sample $\mathbb{Y}=(Y_1, \ldots, Y_n)$ where $Y_i\stackrel{d}{=} Y$, it is clear that the factor $(\beta_1 \vee \beta_2)$ on the right-hand side in (\ref{kret2}) has no effect  on its value, as that statistic is scale invariant. However, it will be affected by $\beta$. In turn, the value of this
$\beta$ will generally be affected by the variances of $X_1$ and $X_2$ as well as their correlation $\rho$. The nature of this dependence is provided in the result below.
%@@
%@@
\begin{Proposition}
\label{beta.fun.prop}
Let ${\bf X} = (X_1, X_2)$ be a random vector with correlation $\rho$ and covariance matrix $\boldsymbol \Sigma=[R_{ij}]_{i,j=1,2}$, where $R_{ii}=Var(X_i)$, $i=1,2$, and $R_{ij} = \rho \sqrt{Var(X_1)Var(X_2)}$ ($i\neq j$). Then, the quantity $\beta$ in (\ref{beta_def}), where $\beta_1$ and $\beta_2$ are the two eigenvalues of $\boldsymbol \Sigma$, is given by
\begin{equation}
\label{fun.h}
\beta = h(\gamma, r) = \frac{1+\gamma - \sqrt{(1-\gamma)^2+4\gamma r}}{1+\gamma + \sqrt{(1-\gamma)^2+4\gamma r}},
\end{equation}
where $r=\rho^2 \in [0,1]$ and
\begin{equation}
\label{gamma_def}
\gamma  =
%\frac{R_{11} \wedge R_{22}}{R_{11} \vee R_{22}}  =
\frac{Var(X_1)\wedge Var(X_2)}{Var(X_1) \vee Var(X_2)}\in [0,1].
\end{equation}
Moreover, the function $h(\gamma, r)$ is strictly increasing in $\gamma\in [0,1]$ for each fixed $r\in [0,1)$ and is strictly decreasing in $r\in [0,1]$ for each fixed $\gamma\in (0,1]$.
\end{Proposition}
%@@
%@@
\begin{proof}
Straightforward algebra shows that the two eigenvectors of $\boldsymbol \Sigma$ are given by
\begin{equation}
\label{lambdas}
    \beta_{1,2} = \frac{1}{2}\left(R_{11} + R_{22} \pm \sqrt{(R_{11}-R_{22})^2+4R_{11}R_{22}\rho^2}\right),
    \end{equation}
so that
\begin{equation}
\label{beta_def1}
\beta  = \frac{\beta_1 \wedge \beta_2}{\beta_1 \vee \beta_2}  =
\frac{R_{11} + R_{22} - \sqrt{(R_{11}-R_{22})^2+4R_{11}R_{22}\rho^2}}{R_{11} + R_{22} + \sqrt{(R_{11}-R_{22})^2+4R_{11}R_{22}\rho^2}}.
\end{equation}
Upon dividing the numerator and the denominator of the expression on the right-hand side in (\ref{beta_def1}) by $R_{11} \vee R_{22}$ we obtain (\ref{fun.h}). The monotonicity properties of the function $h(\gamma,r)$ are established in a standard way by checking that the partial derivatives have appropriate signs.
\end{proof}
%@@
\begin{Remark}
The largest value of $\beta$ is equal to $1 = h(1,0)$ (the variances are equal while the correlation is zero).  The smallest value of $\beta$ is zero, and it occurs at the boundary $\{(0,r) : 0\leq r\leq 1\} \cup \{(\gamma,1) : 0\leq \gamma \leq 1\}$ of the unit square. In particular, we have $\beta=0$ whenever $\rho=\pm 1$, regardless of the ratio of the variances.
By the continuity of the function $h(\gamma, r)$, the value of $\beta$ will be close to zero when the correlation is close to $\pm 1$. If the variances are equal, then we have
\begin{eqnarray}
\label{beta_rho}
\beta = h(1,r) = \frac{1 - |\rho|}{1 + |\rho|},
\end{eqnarray}
and the inverse relation between $\beta$ and $|\rho|$ is quite clear, with small values of $\beta$ for large values of $|\rho|$.
However, the value of $\beta$ can be also close to zero for any value of $\rho$, in particular when the correlation is zero or close to zero and the ratio of the variances is very different than one (so that $\gamma$ is close to zero).
\end{Remark}
%@@
%{\textcolor{red}{Marek ma zmienic oznaczenia 
%\begin{Remark}
% To get an insight about what values of $\gamma$ and $r=\rho^2$ produce which values of $\beta$, in particular small values of $\beta$, close to zero (where the test based on this statistic performs well) we may wish to include a graph showing the level curves $h(\gamma,r)=c$ of the function $h(\gamma, r)$. For any $c\in [0,1]$, these curves consist of the points $(\gamma, r)$ in the unit square that satisfy the equation:
%
%\[
%(1-c)^2((1+\gamma)^2 - (1+c)^2(1-\gamma)^2 = 4(1+c)^2 \gamma r.
%\]
%
%When we solve this equation for $r$, we get:
%
%\[
%r = \frac{1}{2}\left[1+\left(\frac{1-c}{1+c}\right)^2 \right] - \frac{1}{4}\left[1-\left(\frac{1-c}{1+c}\right)^2 \right]\left[ \gamma +\frac{1}{\gamma} \right].
%\]
%
%The above function is monotonically increasing in $\gamma$ on the interval $(0,1]$. We may wish to plot these curves (their common part with the unit square) for various values of $c$, such as $c=0, 0.1, 0.2, \ldots , 0.9, 1$.
%\end{Remark}
%}}
Simulations show that the distribution of the Greenwood statistic (\ref{def_Greenwood1}) is stochastically decreasing with respect to $\beta\in [0,1]$ when it is evaluated on the random sample $Y_1, \ldots, Y_n$ and $Y_i\stackrel{d}{=}Y$ with $Y$ given by (\ref{kret2}). Obviously,  the factor $(\beta_1 \vee \beta_2)$, having to do with the two eigenvalues of the covariance matrix, plays no role here due to invariance property of that statistic with respect to scaling. This plays a crucial role in setting up a rejection region for testing the null hypothesis that $\alpha=2$ based on this statistic, where $\alpha\in (0,2]$ is the tail index of the underlying bivariate sub-Gaussian distribution. In Fig. \ref{fig:S2dist}, we present the CDF of $S_{\vec{\mathbb{X}}}^2$ statistic for a bivariate Gaussian distribution with different values of the $\beta$ parameter (see left panel) and the corresponding right tail (1-CDF) in a log-log scale (right panel). The plots were obtained based on $10,000$ Monte Carlo simulations of bivariate Gaussian random vectors of length $100$. From the plots, we can conclude that the $S_{\vec{\mathbb{X}}}^2$ statistic exhibits stochastic ordering with respect to $\beta$, with its survival function taking its highest value when $\beta=0$ (which corresponds to the case $\rho=\pm 1$).

Note that the discussion presented above, expressed in terms of the $\beta$ parameter, corresponds to the case where $\alpha = 2$. 
When $\alpha < 2$, the $\beta$ parameter also influences the behavior of the $S_{\vec{\mathbb{X}}}^2$ statistic. In this case the simulation study presented in Section \ref{sim:section} clearly confirms the stochastic decrease of the $S_{\vec{\mathbb{X}}}^2$ statistic with respect to  $\alpha $ for given $\beta$, as shown in Fig. \ref{fig:S2rho}. This property of the $S_{\vec{\mathbb{X}}}^2$ statistic is summarized in the following remark.

\begin{Remark} \label{stochorder.prop}
Let $S_{\vec{\mathbb{X}}}^2$ be the statistic defined in (\ref{def_Greenwood1}), where
$\vec{{\mathbb{X}}}=(\mathbf{X}_1,\mathbf{X}_2\ldots, \mathbf{X}_n)$ is a random sample coming from a bivariate sub-Gaussian distribution with index of stability $\alpha \in (0,2]$. Then, based on the results of simulations, for a given $\beta \in [0,1]$, $S_{\vec{\mathbb{X}}}^2$ is observed to be stochastically decreasing with respect to $\alpha$. 
\end{Remark}

\begin{figure}
    \centering
    \includegraphics[width=0.8\linewidth]{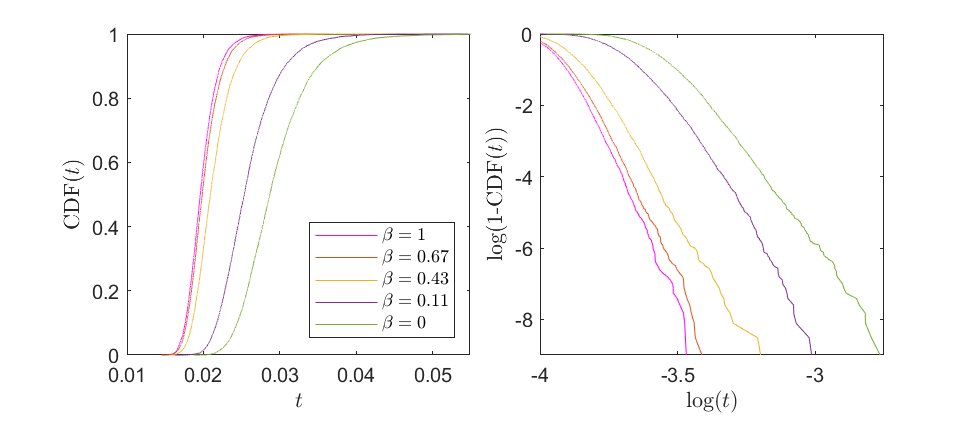}
    \caption{The CDF (left panel) and the tail of the distribution (right panel) of $S_{\vec{\mathbb{X}}}^2$ statistic for bivariate Gaussian distribution with different values of $\beta$ parameters. The plots were obtained based on $10000$ Monte Carlo simulations of bivariate Gaussian random vectors of length $100$.}
    \label{fig:S2dist}
\end{figure}

Based on the theoretical results discussed above and supported by simulations, we present Table \ref{tab1}, which compares the key properties of the two statistics introduced for bivariate sub-Gaussian random vectors. The insights from this comparison serve as a foundation for the testing methodology that will be introduced in the next section. The summary is presented using the parameters $\alpha$ and $\beta$, as the subsequent analysis is also discussed in these terms. The relationship between $\beta$ and $\rho$ is given in Eq. (\ref{fun.h}).

%As a summary of this section,  based on the theoretical results and conducted simulations, we present Table \ref{tab1}, which compares the discussed properties of the two statistics introduced  for bivariate sub-Gaussian random vectors. This table serves as a foundation for the testing methodology introduced in the next section.
%
\begin{table}[htp!]
\centering
\begin{tabular}{c|c|c}
\hline
parameters & $ S_{\vec{\mathbb{X}}}^1$ (Eq. (\ref{def_Greenwood2}))& $ S_{\vec{\mathbb{X}}}^2$ (Eq. (\ref{def_Greenwood1}))  \\ \hline \hline
$\alpha=2$ & independent of $\beta$ & stoch. order wrt $\beta$ \\\hline
$\alpha<2$& independent of $\beta$& no stoch. oder wrt $\beta$  \\\hline
$0\leq \beta \leq 1$&stoch. order wrt $\alpha$ &stoch. order wrt $\alpha$ \\\hline
\end{tabular}
\caption{Comparison of the properties of the considered statistics for bivariate sub-Gaussian random vectors. }\label{tab1}
\end{table}
%

%@@
%@@
%@@
\section{Testing methodology}
\label{sec:testing}
%@@
%@@
%@@
%In this section, we present the proposed application of modified Greenwood statistic in the goodness-of-fit testing. At first, we discuss the application of the  statistic in testing for univariate symmetric $\alpha-$stable distributions and also in the estimation of the parameter $\alpha$. Secondly, we extend the proposed methodology to multivariate case for the class of sub-Gaussian distributed vectors.
In this section, we discuss the use of the modified Greenwood statistic for tests about $\alpha$ within the $\alpha$-stable family of distributions in both univariate and multivariate cases. In the univariate case we also present the estimation method for $\alpha$ parameter.
%.  of univariate and multivariate distributions. We first introduce the methodology for univariate symmetric $\alpha$-stable distributions, focusing on both testing and parameter estimation. We then extend this approach to the multivariate case, specifically for the class of sub-Gaussian distributed vectors.
%@@
%@@
\subsection{The univariate case}
\label{int}
%@@
%@@

%The Greenwood statistic can be used for testing whether a given random sample originates from a Gaussian distribution within the framework of $\alpha$-stable distributions.
As discussed in  \cite{skowronek2024}, the Greenwood statistic can be used to  test for Gaussianity within the class of $\alpha$-stable distributions. The problem reduces to testing $\alpha=2$ (Gaussian distribution) versus $\alpha < 2$ (heavy tailed $\alpha$-stable distribution. Given a random sample $\vec{\mathbb{X}}$ from an $\alpha$-stable model, we focus on testing  the following hypotheses
\begin{equation}\label{hyp1}
    \mathcal{H}_0 \text{ : } \alpha = 2 \text{, } \mathcal{H}_1 \text{ : } \alpha < 2.
\end{equation}
\noindent We propose $S_{\mathbb{X}}$,  defined in (\ref{def_Greenwood}) as the test statistic.  Since $S_{\mathbb{X}}$, decreases stochastically with respect to  $\alpha$ (see, Proposition 3 (ii) in \cite{arendarczyk2022}) we can construct the rejection region as the interval

\begin{equation}\label{rejection1}
    [\hat{Q}_{1-c}(n),1],
\end{equation}
\noindent where $n$ is the sample size, $\hat{Q}_{1-c}(n)$ is the $(1-c)$th quantile of the distribution of $S_{{\mathbb{X}}}$ under the null hypothesis, and $c$ is the level of significance. The quantile $\hat{Q}_c(n)$ can be estimated via Monte Carlo simulations.

%The test rejects $\mathcal{H}_0$ if the observed value of $S_{\vec{\mathbb{X}}}$ falls within the rejection region.  
%This test can also be interpreted as a test for finite variance within the class of $\alpha$-stable distributions, where $\alpha=2$ indicates finite variance, and $0<\alpha<2$ implies infinite variance. Therefore, rejecting $\mathcal{H}_0$ would suggest that the sample has infinite variance, as discussed in  \cite{skowronek2024}.

The test  described above can be generalized to a more general problem of testing 
%broader testing scenario, where we are interested in determining whether a random sample originates from an $\alpha$-stable distribution with a stability index $\alpha$ greater or equal to a specific value $\alpha^* \in (0,2]$.  The hypotheses of the test are given by
%
\begin{equation}\label{test2}
    \mathcal{H}_0 \text{ : } \alpha \geq \alpha^* \text{, } \mathcal{H}_1 \text{ : } \alpha < \alpha^*
\end{equation}
for any $\alpha^* \in (0,2]$. Again, the test relies on the stochastic monotonicity of $S_{\mathbb{X}}$  with respect to $\alpha$ within the class of symmetric $\alpha$-stable distributions. The rejection region in this case can also be estimated using Monte Carlo simulation. 
%it can be used as the test statistic in this more general scenario as well. Similar to the test for Gaussianity, we employ a one-sided rejection region defined as in (\ref{rejection1}). However, in this case, the critical quantile $\hat{Q}_c(n)$ is derived from the distribution of $S_{\vec{\mathbb{X}}}$ based on samples of length $n$ from $S(\alpha^{*},1)$ distribution. Thus, the test rejects the null hypothesis $\mathcal{H}_0$ if the observed value of $S_{\vec{\mathbb{X}}}$ falls within the one-sided rejection region, indicating that the stability index $\alpha$ is likely to be less than $\alpha^*$. 
The same rejection region can also be used for testing 
\begin{equation}\label{rightsided2}
    \mathcal{H}_0 \text{ : } \alpha = \alpha^* \text{, } \mathcal{H}_1 \text{ : } \alpha < \alpha^*.
\end{equation}
In addition, we can use this methodology for testing 
\begin{equation} \label{leftsided_H}
    \mathcal{H}_0 \text{ : } \alpha \leq \alpha^* \text{, } \mathcal{H}_1 \text{ : } \alpha > \alpha^* \ \ \ {\rm or} \ \ \
    \mathcal{H}_0 \text{ : } \alpha = \alpha^* \text{, } \mathcal{H}_1 \text{ : } \alpha > \alpha^*.
\end{equation}
In this case, the rejection region is the following interval
\begin{equation} \label{leftsided}
    \left[\frac{1}{n}, Q_c(n)\right].
\end{equation}
Finally, we can  use statistic $S_{\mathbb{X}}$ for constructing two-sided tests for $\alpha$
\begin{eqnarray} \label{twosided_H}
    \mathcal{H}_0 \text{ : } \alpha = \alpha^* \text{, } \mathcal{H}_1 \text{ : } \alpha \neq \alpha^*,
\end{eqnarray}
with the rejection region for the following form
\begin{eqnarray} \label{twosided}
    \left[\frac{1}{n}, Q_{c/2}(n)\right] \cup \left[Q_{1-c/2}(n), 1\right],
\end{eqnarray}
where the quantiles (of the $S_{{\mathbb{X}}}$ statistic distribution) in (\ref{leftsided}) and (\ref{twosided}) are  received based on Monte Carlo simulations of the samples of size $n$ from $S(\alpha^{*},1)$ distribution.

{
Since the statistic $S_{{\mathbb{X}}}$ is stochastically decreasing with respect to $\alpha \in (0,2]$, it follows that $\mathbb{P}\left(S_{{\mathbb{X}}}\in {\rm C}\right) \leq c$, where ${\rm C}$ is the appropriate rejection region, for all the considered tests.  Moreover, the power of each test increases as the stability index $\alpha$ moves away from the set of $\alpha$'s under $\mathcal{H}_0$. This is summarized in the result below.
}
%@@
%@@
\begin{Proposition}
All the tests based on the statistic $S_{{\mathbb{X}}}$ defined in (\ref{def_Greenwood}) and described above
for the hypotheses specified in (\ref{hyp1}), (\ref{test2}), (\ref{rightsided2}), (\ref{leftsided_H}), and (\ref{twosided_H}) have size $c$ and are unbiased.
\end{Proposition}
%@@
%@@

The testing methodology discussed above can be used to obtain one and two sided confidence sets of the stability index
$\alpha$ via the standard process of inverting the test. More precisely, for a given random sample ${\mathbb{X}}$ from a symmetric $\alpha$-stable distribution, the confidence set consists of those values of $\alpha^*$ for which the relevant test does not reject the null hypothesis with the given data. By the stochastic monotonicity of the relevant test statistic with respect to the parameter $\alpha$, these confidence sets are actually intervals, and their exact form can be obtained by Monte Carlo simulations.

%The testing methodology discussed above can also be viewed as the estimation of the $\alpha$ parameters for a random sample from the symmetric $\alpha$-stable distribution. More precisely, for given a random sample $\vec{\mathbb{X}}$, we perform the testing procedure for the considered hypotheses, for each $\alpha^{*}$ value from the set $\alpha^{**} = \{\alpha_1^{*}, \alpha_2^{*}, \ldots, \alpha_K^{*}\}$. For each value $\alpha^{*}$ from this set we separately construct the rejection region  and check whether the value of the test statistic $S_{\vec{\mathbb{X}}}$ calculated for the random sample falls within the constructed rejection region. The confidence interval of $\alpha$ corresponds to those values of $\alpha^{*}$ in the set $\alpha^{**}$ for which the value of $S_{\vec{\mathbb{X}}}$ lies outside the corresponding rejection region.

%@@
%@@
\subsection{The bivariate case}
\label{biv}
%@@
%@@
%In the bivariate case, we proceed in the same way as in the univariate scenario. The only difference lies in the test statistic applied. Specifically, to test whether a bivariate random sample comes from a bivariate Gaussian distribution (which corresponds to the $\mathcal{H}_0$ hypotheses defined in (\ref{hyp1}) for the univariate scenario), we can use either the statistic $S_{\vec{\mathbb{X}}}^1$ or the statistic $S_{\vec{\mathbb{X}}}^2$,  defined in Eq. (\ref{def_Greenwood2}) and Eq. (\ref{def_Greenwood1}), respectively. However, the choice of the test statistic will affect the construction of the rejection region (\ref{rejection1}). For the $S_{\vec{\mathbb{X}}}^1$ statistic, the rejection region is constructed based on $M$ Monte Carlo simulations from a bivariate Gaussian-distributed samples of length $n$. These samples can be simulated for any value of the $\rho$ parameter, as it does not affect the distribution of the test statistic. However, when considering the $S_{\vec{\mathbb{X}}}^2$ statistic, the rejection region (\ref{rejection1}) is constructed using samples from the bivariate Gaussian distribution with $\rho=1$. In consequence, the appropriate quantile used in the rejection region is calculated from the $M$ values of the $S_{\vec{\mathbb{X}}}$ statistic calculated for the random sample from the $\chi^2$ distribution with one degree of freedom.

The approach to testing in the bivariate case is similar to the one for the univariate problem. The main difference is the test statistic. When testing bivariate Gaussianity the hypotheses are the same as those given in Section \ref{int}. However, we have a  choice between two test statistics.  Specifically, we can use either statistic $S_{\vec{\mathbb{X}}}^1$ or  statistic $S_{\vec{\mathbb{X}}}^2$,  which are defined in Eq. (\ref{def_Greenwood2}) and Eq. (\ref{def_Greenwood1}), respectively. The choice of the statistic affects the construction of the  rejection region.

\begin{itemize}

\item Test statistic is  $S_{\vec{\mathbb{X}}}^1$: The rejection region is constructed based on Monte Carlo simulations of samples of length $n$ from a bivariate Gaussian distribution with any correlation structure, because the distribution of  $S_{\vec{\mathbb{X}}}^1$ is not sensitive to the correlation structure. 

\item Test statistic is $S_{\vec{\mathbb{X}}}^2$: The rejection region is based on samples from the bivariate Gaussian distribution with $\beta=0$ (or $\rho=1$).  In this case, the  quantile needed for the  rejection region is calculated using simulated values of the  $S_{\vec{\mathbb{X}}}^2$ computed for samples from the $\chi^2$ distribution with one degree of freedom.

\end{itemize}

In the case where the $\mathcal{H}_0$ and $\mathcal{H}_1$ hypotheses are defined as in (\ref{test2}), (\ref{rightsided2}), (\ref{leftsided_H}), and (\ref{twosided_H}), we propose using $S_{\vec{\mathbb{X}}}^1$ as the test statistic, with the rejection region constructed based on  Monte Carlo simulated samples from a bivariate sub-Gaussian distribution with parameter $\alpha^{*}$. Similar to testing bivariate Gaussianity, the selection of the $\rho$ parameter (used to construct the rejection region) is not important, as this statistic is independent of this parameter.

Both statistics, $S_{\vec{\mathbb{X}}}^1$ and $S_{\vec{\mathbb{X}}}^2$, can also be applied when testing for a bivariate Gaussian distribution (corresponding to the hypotheses defined in (\ref{hyp1})) and a bivariate sub-Gaussian distribution with a specific value of the parameter $\alpha$ (corresponding to the hypotheses defined in (\ref{test2}), (\ref{rightsided2}), (\ref{leftsided_H}), and (\ref{twosided_H})) under the assumption that the $\rho$ parameter is known. In this case, the rejection regions are defined similarly to the cases discussed previously. However, the appropriate quantile in (\ref{rejection1}) is calculated for the bivariate Gaussian or sub-Gaussian (with the tested $\alpha^*$ parameter) random samples with the given $\rho$ parameter. While this approach has less practical importance compared to the cases discussed earlier, one can perform the testing procedure separately for $\rho$ parameters from a given set.

%@@
%@@
%@@
\section{Simulation study}
\label{sim:section}
%@@
%@@
%@@
In this section, we present results  on the power of the tests described in the previous section and their comparison with other tests for bivariate Gaussianity (i.e. for sub-Gaussian distribution with $\alpha=2$) based on $S_{\vec{\mathbb{X}}}^1$ and $S_{\vec{\mathbb{X}}}^2$ statistics. The simulations were performed in MATLAB version 23.2.0.2365128 (R2023b). We consider  the following hypotheses:  $\mathcal{H}_0$ - a  random sample comes from a bivariate Gaussian distribution, $\mathcal{H}_1$- a random sample comes from bivariate sub-Gaussian distribution with $\alpha < 2$. The rejection regions for both $S_{\vec{\mathbb{X}}}^1$ and $S_{\vec{\mathbb{X}}}^2$ test statistics are calculated based on $10,000$ Monte Carlo simulations of the bivariate random samples under the null hypothesis. For $S_{\vec{\mathbb{X}}}^1$ and $S_{\vec{\mathbb{X}}}^2$, we  calculate the rejection regions based on bivariate Gaussian distribution with $\rho=1$ (i.e. $\beta=0$, see Eq. (\ref{beta_rho})). Although  we present only selected illustrative examples, we applied our methods to samples with a wider range of parameters with similar results.

In order to demonstrate the impact of  parameter $\beta$ on the $S_{\vec{\mathbb{X}}}^2$-based test,  we present the power of the test for different values of  $\beta$   in $\mathcal{H}_1$ in Fig.~\ref{fig:S2rho}. Since  we calculate the critical region for $\alpha=2$ and $\beta=0$,  the power of a test may vary for different values of $\beta$.  For each considered $\alpha$, the test based on $S_{\vec{\mathbb{X}}}^2$ had the highest power  for $\beta=0$ and thus $\beta=0$ was used to calculate the critical region. This follows from the stochastic order of statistic $S_{\vec{\mathbb{X}}}^2$  with respect to $\beta$ discussed in  Section \ref{sec:greenwood}. Finally, recall that statistic $S_{\vec{\mathbb{X}}}^1$ is independent on $\rho$ (and thus $\beta$). 

\begin{figure}
    \centering
    \includegraphics[width=0.8\linewidth]{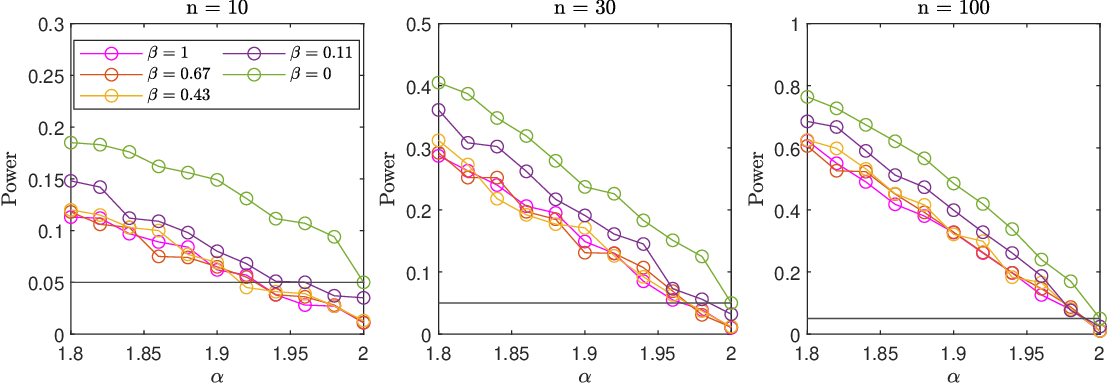}
    \caption{Comparison of the power curves of a  test based on $S_{\vec{\mathbb{X}}}^2$ statistic obtained for  different values of $\beta$ parameters (see Eq. (\ref{beta_rho})) for sample sizes $n = 10,30,100$. The power was obtained based on $1000$ Monte Carlo simulated samples corresponding to $\mathcal{H}_1$ hypothesis. The critical region was constructed based on $10000$ simulated samples from bivariate Gaussian distribution with $\rho=1$ (i.e. $\beta=0$).}
    \label{fig:S2rho}
\end{figure}

The power of tests based on $S_{\vec{\mathbb{X}}}^1$ and $S_{\vec{\mathbb{X}}}^2$ statistics we compare with the power of the following tests for multivariate Gaussianity considered in the literature: { \textcolor{black}{ Jarque-Bera (JB) test~\cite{jb2test}, Henze-Zirkler (HZ) test~\cite{henzezirkler}, Mardia's kurtosis (Kurt) and skewness (Skew) tests~\cite{mardia}, and  Royston (Royst) test~\cite{royston}.   To estimate the power  we simulated $1,000$ random samples under the alternative hypothesis.  Under $\mathcal{H}_1$, we selected stability indexes in the range $\alpha \in \{1.8, 1.82, \ldots, 1.98, 2\}$. We present the results for sample sizes $n=10,30$ and $n=100$ and significance level $c=0.05$. In the remainder of this paper  we denote tests based on statistics $S_{\vec{\mathbb{X}}}^1$ and $S_{\vec{\mathbb{X}}}^2$ as $S1$ test and $S2$ test, respectively.}

The power comparison of the $S1$ and $S2$ tests with the other tests  is presented in Fig.~\ref{fig:S1literature} and Fig.~\ref{fig:S2literature}.  We observe that the $S1$ test performs better than other tests  for $n=10$ (see Fig.~\ref{fig:S1literature}). For $n=30$ and $\alpha < 1.86$,  Kurtosis test, Skewness test and Royston test outperform the  $S1$ test. However, as stability index approaches  $2$ (i.e. the bivariate sub-Gaussian  distribution approaches  the Gaussian law) the $S1$ test outperforms all other tests considered in this comparison. For $n=100$, the $S1$ test is outperformed by Kurtosis test and Royston tests. Thus, we conclude that $S1$ test is more powerful than the other tests we considered for small sample sizes ($n<30$), and  in particular when $\alpha$ close to $2$. 

 Since the power of $S2$ test  depends on the value of $\beta$, in Fig.~\ref{fig:S2literature} we present the power of $S2$  for three different values of $\beta$, namely  $\beta\in\{0.005,0.051,0.081\}$ (corresponding to  $\rho \in \{0.99, 0.9, 0.85\}$).  For $\beta=0.005$ and $n=10$, and $\beta \in\{0.052,0.081\}$ with $n = 30$, test $S2$ outperforms the $S1$ and all other considered tests, especially when $\alpha$ is close to 2. For $\beta = 0.081$ and $n = 10, 30$, as $\alpha$ approaches $2$, the $S1$ test is the most powerful among all the tests considered. For $n = 100$, the $S1$ and $S2$ tests are outperformed by the Kurtosis and Royston tests. Although we do not  present results for larger $\beta$'s  (i.e. smaller $\rho$), awe found that in those cases, the $S2$ test is less powerful than the $S1$ test and all other considered tests.  Based on our findings, we recommend using the $S1$ test (instead of the $S2$ test) when the correlation  between the data is unknown and the number of observations is small.

\begin{figure}
    \centering
    \includegraphics[width=0.8\linewidth]{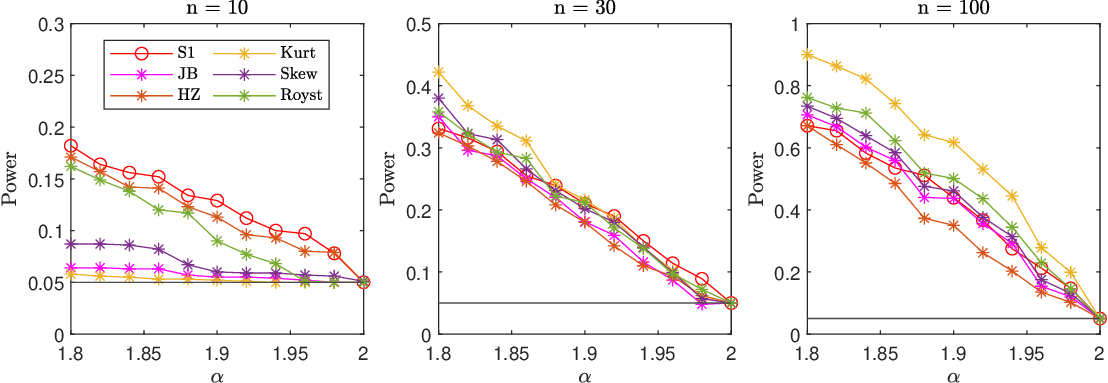}
    \caption{Comparison of the power of a test  based on $S_{\vec{\mathbb{X}}}^1$ statistic with other tests known in the literature (JB, HZ, Kurt, Skew and Royst test). The power curves were obtained based on $1000$ Monte Carlo simulated samples corresponding to $\mathcal{H}_1$ hypothesis. The critical region was constructed based on $10000$ simulated samples from bivariate Gaussian distribution.}
    \label{fig:S1literature}
\end{figure}

\begin{figure}
    \centering
    \includegraphics[width=0.8\linewidth]{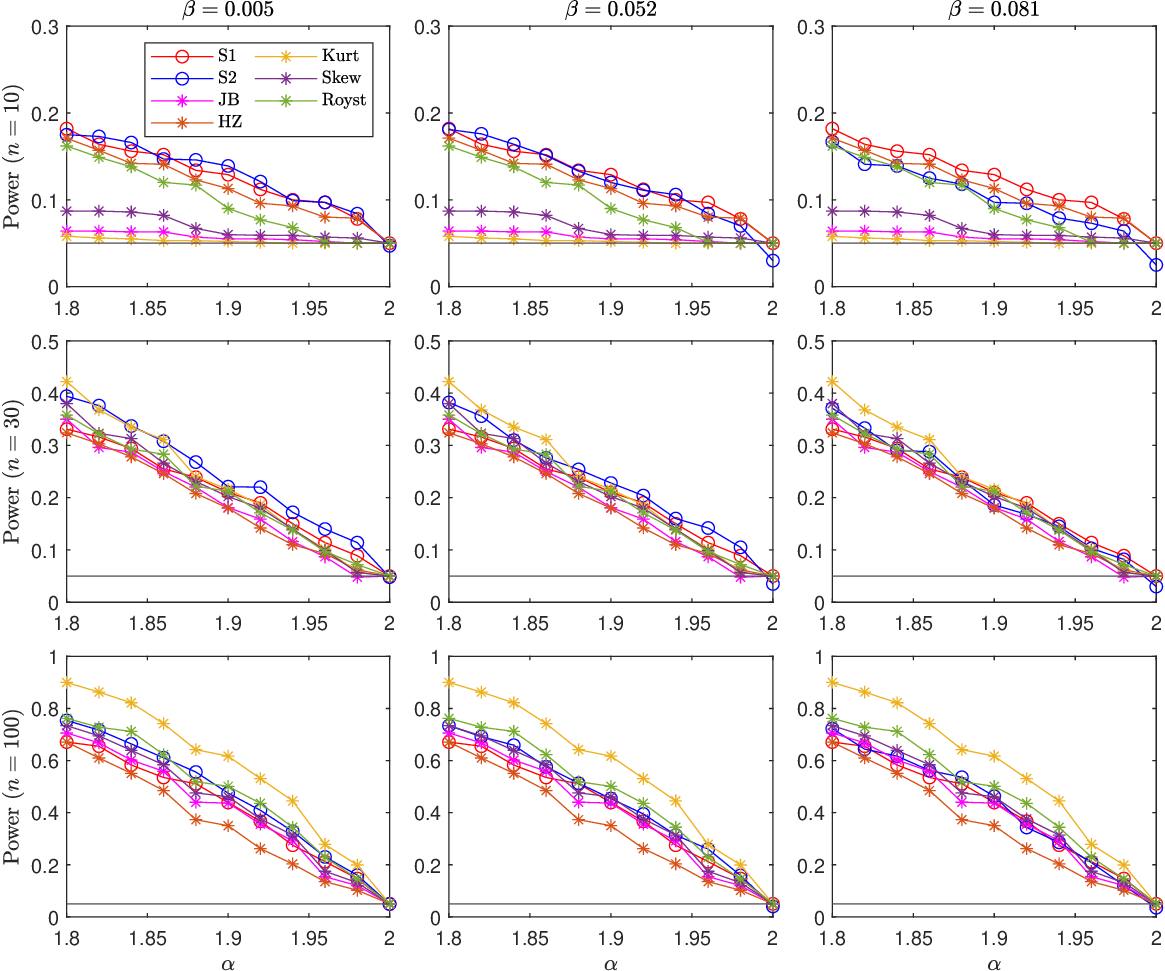}
    \caption{Comparison of the power of tests  for for bivariate Gaussian distribution based on $S_{\vec{\mathbb{X}}}^1$ and $S_{\vec{\mathbb{X}}}^2$ statistic obtained for different values of $\beta$ parameter  with other tests known in the literature (JB, HZ, Kurt, Skew and Royst test). The power curves were obtained based on $1000$ Monte Carlo simulated samples corresponding to $\mathcal{H}_1$ hypothesis. The critical region was constructed based on $10000$ simulated samples from bivariate Gaussian distribution with $\rho=1$ (i.e. with $\beta=0$).}
    \label{fig:S2literature}
\end{figure}

%@@
%@@
%@@
\section{Data analysis}
\label{sec:real}
%@@
%@@
%@@
In this section, we apply our methodology to analyze the financial data describing the main risk factors considered by a copper mining company (KGHM). The main question is whether the data comes from heavy tail or Gaussian models.  The dataset was introduced and discussed in~\cite{CUprice}, where the authors proposed a bivariate time series model based on symmetric $\alpha-$stable distribution. This work concentrates on the analysis of the residuals for that model. 

\subsection{Univariate case}

Our  bivariate dataset~\cite{CUprice} represents a vector of observations corresponding to USDPLN exchange rate and CU (copper) price (in USD). The  data are not homogeneous in time, which means that the behavior  of the data changes with time. To obtain reasonably homogeneous data sets for analysis, we divided the full data into two homogeneous subsets representing different market regimes using a hidden Markov model for $\alpha-$stable distribution (see, e.g., ~\cite{Hmm_rabiner}, \cite{CUprice}).   This segmentation resulted in following two  time frames: (1) the 2008 crisis regime (R1) from March 2006 to the first half of 2012, and (2) a stable market regime (R2) from the second half of 2012 to second half of 2020. Regimes R1 and R2 consist of $n=335$ and $n=423$ observations, respectively. We use the results of ~\cite{CUprice} where the authors  fitted the following VAR-type time series model of order $1$ with symmetric $\alpha$-stable innovations
\begin{equation}\label{VAR}
   \mathbf{X}_t = M \mathbf{X}_{t-1} + \mathbf{\xi}_t, ~~t\in \mathbb{Z},
\end{equation}
where $\mathbf{X}_t=\left(X_1^{(t)},X_2^{(t)}\right)$ and $\mathbf{\xi}_t=\left(\xi_1^{(t)},\xi_2^{(t)}\right)$ is a bivariate symmetric $\alpha$-stable residual series. In the model equation (\ref{VAR})  $M$ is a $2\times 2$ matrix parameter.  The estimates of $M$ were  obtained by a modified Yule-Walker approach~\cite{Grzesiek2021} with results as follows
\begin{equation}
    \hat{M}_{R1} =
    \begin{bmatrix}
        0.2927 & 0\\
        0 & 0.2100
    \end{bmatrix},
    \quad
    \hat{M}_{R2} =
    \begin{bmatrix}
        0.281 & 0\\
        0 & 0.1386
    \end{bmatrix}.
\end{equation}
The diagonal form of the $\hat{M}$ matrices indicates that  the components of the model (\ref{VAR}) are dependent only through the dependence of the residual series, i.e., $\xi_1^{(t)}$ and $\xi_2^{(t)}$. Given the estimated parameters of the process, we obtained the time series of residuals for both processes presented in Fig.~\ref{fig:CU}:  USDPLN exchange rates (left panel),  CU price (right panel).

\begin{figure}[htp!]
   \centering
   \includegraphics[width=0.45\textwidth]{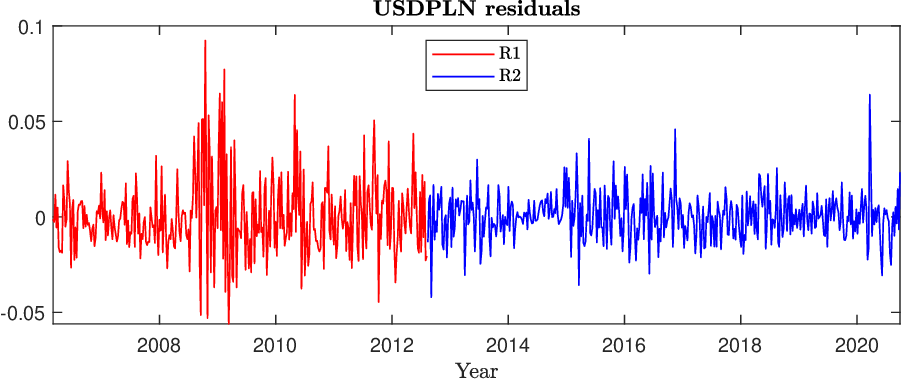}%\\
 %  \vspace{0.5cm}
    \includegraphics[width=0.45\textwidth]{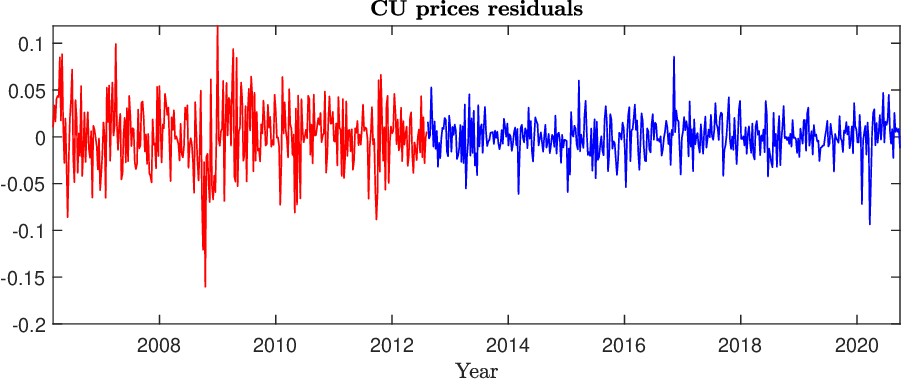}
    \caption{The residuals of analyzed financial dataset after using VAR(1) model with symmetric $\alpha-$stable distribution. Left panel represents residuals of USDPLN exchange rate from 2006 to 2020 and the right panel represents residuals of CU prices from 2006 to 2020. Red color presents regime R1 (2008 crisis) and blue colour represents regime R2 (stable market situation).}
    \label{fig:CU}
\end{figure}
To estimate $\alpha$ we used the results of ~\cite{CUprice} where the estimates $\hat{\alpha} = 1.7229$ for R1 and $\hat{\alpha} = 1.8424$ for R2 (for USDPLN) and $\hat{\alpha} = 1.9219$ for R1 and $\hat{\alpha} = 1.8243$ for R2 (for CU price). Next, we computed the 95\% confidence intervals for $\alpha$ by the standard "test inversion" method described in section \ref{sec:testing}. The results are presented in Table \ref{tab2}.

\begin{table}[htp!]
\centering
\begin{tabular}{c|c|c}
\hline
regime & USDPLN& CU price  \\ \hline \hline
R1 & [1.69,1.98] &  [1.76,1.99] \\\hline
R2 &[1.75,1.99]& [1.72,1.99] \\\hline
\end{tabular}
\caption{95\% confidence interval for $\alpha$ for USDPLN and CU price for R1 and R2 regimes. The results are obtained based on $1000$ Monte Carlo simulations of the random samples from symmetric $\alpha-$stable distribution. }\label{tab2}
\end{table}

The confidence intervals support the conclusion that the USDPLN and CU prices follow heavy tail $\alpha-$stable distributions ($\alpha<2$) and confirm the conclusions of analysis presented in~\cite{CUprice}.  
\subsection{Bivariate case}

In this section, we apply tests $S1$ and $S2$ for bivariate Gaussian distribution within the class of sub-Gaussian distributions to the bivariate data with marginals analyzed above. More precisely, we treat the residual series presented in Fig. \ref{fig:CU} as a bivariate vector. We consider residuals from R1 and  R2 regimes separately.  In addition to tests $S1$ and $S2$ we performed the Mardia's kurtosis test~\cite{mardia}, which is thought of as one of the best tests for multivariate normality.  In case of $S2$ test, we needed to standardize the univariate data (corresponding to separate components and regimes), and we accomplish this using the standardization approach based on conditional standard deviation described in \cite{maraj2023}.
The values of the test statistics and the critical values of the tests for bivariate Gaussian distribution, are presented in Table \ref{tab3}. The  critical values were estimated based on $10,000$ Monte Carlo simulated random samples from bivariate Gaussian distribution with $\rho=1$.

\begin{table}[htp!]
\centering
\begin{tabular}{l|l|l|l}
\hline
regime & Kurt test & S1 test & S2 test \\ \hline\hline
R1     &  $8.99 \text{  } (1.57)$         &  0.0052 $(0.0049)$        &   0.0114 $(0.0103)$   \\ \hline
R2     & $ 14.06  \text{  }  (1.63)$          & 0.0043 $(0.0039)$        & 0.0102 $(0.0081)$        \\ \hline
\end{tabular}
\caption{The values of the test statistics (calculated for the residual series from regimes R1 and R2) and the critical values (in parenthesis) of the tests (Kurt test, S1 test and S2 test) for the bivariate Gaussian distribution. The critical values are calculated based on $10000$ Monte Carlo simulated samples coming from bivariate Gaussian distribution with $\rho=1$.}\label{tab3}
\end{table}

All tests reject the bivariate normality, thus our conclusion is that the residual series for both regimes follow a heavy tail symmetric $\alpha-$stable distribution ($\alpha < 2$). 

\section{Conclusions}
In this paper, we present two statistical tests based on the  modified Greenwood statistic introduced in \cite{skowronek2024}, for testing multivariate $\alpha$-stable distributions. Our  focus is on a specific sub-class of $\alpha$-stable distributions, namely the sub-Gaussian distributions, within the bivariate context. The classical Greenwood statistic along with its modified version for univariate  random samples has been extensively applied across various fields. A key feature of both the classical and modified Greenwood statistics is their stochastic monotonicity within the class of univariate symmetric $\alpha$-stable distributions, a property that renders these statistics suitable as test statistics for hypotheses within this broad family of distribution. Extending the modified Greenwood statistic to multivariate random vectors facilitates testing of multivariate $\alpha$-stable distributions, particularly the multivariate Gaussian model.  A common challenge found in the literature, lies in distinguishing between Gaussian and the heavy tailed $\alpha$-stable distributions, especially when the stability parameter $\alpha$ approaches 2 (i.e., when the $\alpha-$stable distribution is very close  to Gaussian). This issue is prominent in both univariate and multivariate cases. The methodology introduced in this article offers a solution for addressing this problem. Through Monte Carlo simulations we demonstrate the performance of the testing processes, with particular attention to the cases where  $\alpha$ is close to  2 and the sample sizes are small. We also compare the power of the proposed methods with common tests for multivariate Gaussianity.   The practical value of our methods is highlighted by analysis of data  on financial risk factors. However,  our approach can also be applied to data from other fields, such as vibration-based condition monitoring, see e.g. \cite{skowronek2024,ecem2023,SKOWRONEK2023110465}.
\section*{Acknowledgements}
The work of K.S. and A.W. is supported by National Center of Science under Sheng2 project No. UMO-2021/40/Q/ST8/00024 "NonGauMech - New methods of processing non-stationary signals (identification, segmentation, extraction, modeling) with non-Gaussian characteristics for the purpose of monitoring complex mechanical structures".
This paper was partially written during A.K.P.'s and T.J.K.'s visit to the Mathematical Institute at the University of Wroclaw, Poland, in the summer of 2024. The authors express their gratitude to the Institute for providing accommodations during their stay. T.J.K.'s travel was partially supported by the UNR VPRI travel grant, and A.K.P.'s travel was partially supported by the Department of Mathematics \& Statistics at the University of Nevada, Reno.

\bibliography{mybibliography}

\end{document}